\newcommand{\krstyle}[1]{}
\newcommand{\nonkrstyle}[1]{#1}

\krstyle{
	\documentclass{article}
	\pdfpagewidth=8.5in
	\pdfpageheight=11in

	\usepackage{kr}

	\usepackage{soul}
	\usepackage[hidelinks]{hyperref}
	\usepackage[small]{caption}
	\usepackage{graphicx}
	\usepackage{booktabs}
	\usepackage{algorithmic}
	\urlstyle{same}
	\bibliographystyle{kr}
	\usepackage{amsthm}

	\usepackage{times}

	\newtheorem{theorem}{Theorem}
	\newtheorem{lemma}{Lemma}
	\newtheorem{definition}{Definition}
	\newtheorem{example}{Example}
	\newtheorem{proposition}{Proposition}
	\newtheorem{corollary}{Corollary}
	
	\author{%
	Jorge Fandinno$^1$\and
	Markus Hecher$^{1,2}$\\
	\affiliations
	$^1$University of Potsdam, Germany\\
	$^2$TU Wien, Austria\\
	\emails
	\{jorgefandinno, mhecher\}@gmail.com
	}
}

\nonkrstyle{
	\documentclass[a4paper]{article}
	\usepackage[numbers,sort&compress]{natbib}
	\bibliographystyle{abbrvnat}
	\usepackage{authblk}
	\author[1]{Jorge Fandinno\thanks{jorgefandinno@gmail.com}}

	\author[1,2]{Markus Hecher\thanks{hecher@dbai.tuwien.ac.at}}
	\affil[1]{University of Potsdam, Germany}
	\affil[2]{TU Wien, Austria}

	\usepackage{amsthm}
	\newtheorem{theorem}{Theorem}
	\newtheorem{lemma}{Lemma}
	\newtheorem{definition}{Definition}
	\newtheorem{example}{Example}
	\newtheorem{proposition}{Proposition}
	\newtheorem{corollary}{Corollary}

}

\usepackage[utf8]{inputenc}
\usepackage{amsmath}
\usepackage{amssymb}
\usepackage{microtype}
\usepackage{url}

\usepackage{tikz}
\usetikzlibrary{shapes}
\usetikzlibrary{calc}
\usetikzlibrary{positioning}
\usepackage{url}

\usepackage
{xcolor}
\newcommand{\tuplecolor}[1]{\textcolor{#1}}
\newcommand{\inputPredColor}{orange!55!red}
\newcommand{\outputPredColor}{blue!45!black}
\newcommand{\statePredColor}{green!62!black}

\usepackage{mathtools}
\DeclarePairedDelimiter\ceil{\lceil}{\rceil}

\tikzstyle{tdnode} = [draw,rounded corners,top color=vertexTopColor,bottom color=vertexBottomColor,minimum size=1.5em]
\tikzstyle{stdnode} = [tdnode, font=\scriptsize]
\tikzstyle{stdnodecompact} = [stdnode, inner sep = 1.5pt, outer sep = 0.1pt]
\tikzstyle{stdnodetable} = [stdnode, inner sep = 0.5pt, outer sep = 0]
\tikzstyle{stdnodenum} = [minimum size=1.5em, font=\scriptsize]
\tikzstyle{tdedge} = [-,draw,thick]
\tikzstyle{tdlabel} = [draw=none, rectangle, fill=none, inner sep=0pt, font=\scriptsize]
\colorlet{vertexTopColor}{white}
\colorlet{vertexBottomColor}{black!10}

\usepackage[ruled,vlined,linesnumbered]{algorithm2e}
\renewcommand*{\algorithmcfname}{Listing}
\SetKwInput{KwData}{In}
\SetKwInput{KwResult}{Out}
\SetAlFnt{\small}
\SetAlCapFnt{\small}
\SetAlCapNameFnt{\small}
\SetAlCapHSkip{0pt}
\SetEndCharOfAlgoLine{}
\makeatletter
\newcommand{\algorithmfootnote}[2][\footnotesize]{
  \let\old@algocf@finish\@algocf@finish
  \def\@algocf@finish{\old@algocf@finish
    \leavevmode\rlap{\begin{minipage}{\linewidth}
    #1#2
    \end{minipage}}
  }
}
\makeatother

\usepackage{xspace}

\newenvironment{restatetheorem}[1][\unskip]{%
  \begingroup
  \def\thetheorem{\ref{#1}} 
}%
{%
  \addtocounter{theorem}{-1}
  \endgroup
}%
{%
  \addtocounter{lemma}{-1}
  \endgroup
}%

\DeclareMathOperator{\scc}{scc}

\newcommand{\futuresketch}[1]{}
\newcommand{\longversion}[1]{}
\newcommand{\shortversion}[1]{#1}
\DeclareMathOperator{\poly}{poly}
\newcommand{\eqdef}{\ensuremath{\,\mathrel{\mathop:}=}}
\newcommand{\SB}{\{}%
\newcommand{\SM}{\mid}%
\newcommand{\SE}{\}}%

\newcommand{\Card}[1]{\left|#1\right|}
\newcommand{\at}{\text{\normalfont at}}
\newcommand{\por}{\vee}
\newcommand{\hsep}{\leftarrow\,}
\newcommand{\algo}[1]{\ensuremath{\mathsf{#1}}}
\newcommand{\PRIM}{\ensuremath{{\algo{BndCyc}}}\xspace}
\DeclareMathOperator{\width}{width}
\DeclareMathOperator{\children}{children}
\DeclareMathOperator{\rootOf}{root}

\newcommand{\tab}[1]{\ensuremath{\tau_{#1}}}
\DeclareMathOperator{\type}{type}
\newcommand{\intr}{\textit{int}}
\newcommand{\leaf}{\textit{leaf}}
\newcommand{\rem}{\textit{forget}}
\newcommand{\join}{\textit{join}}
\newcommand{\tw}[1]{\mathit{tw}(#1)}
\newcommand{\TTT}{\ensuremath{\mathcal{T}}}%
\newcommand{\MAI}[2]{\ensuremath{#1^+_{#2}}}%
\newcommand{\MAR}[2]{\ensuremath{#1^-_{#2}}}%
\newcommand{\MARR}[2]{\ensuremath{#1^\sim_{#2}}}%
\DeclareMathOperator{\gatherproof}{\algo{proven}}
\DeclareMathOperator{\possord}{\algo{levelMaps}}
\DeclareMathOperator{\isminimal}{\algo{isMin}}
\newcommand{\bvali}[3]{\ensuremath{[\![#1]\!]_{#3}}}
\newcommand{\bval}[2]{\ensuremath{[#1]^{#2}}}

\newcommand{\SAT}{\textsc{SAT}\xspace}
\newcommand{\ASP}{\textsc{ASP}\xspace}

\newcommand{\NP}{\ensuremath{\textsc{NP}}\xspace}
\newcommand{\SIGMA}[2]{\ensuremath{\Sigma_{\textrm{#1}}^{\textrm{#2}}}}

\newcommand{\footnoteitext}[1]{\stepcounter{footnote}
  \footnotetext[\thefootnote]{#1}}

\begin{document}

\title{Treewidth-Aware Complexity in ASP:\\Not all Positive Cycles are Equally Hard}

\maketitle

\begin{abstract}
It is well-know that deciding consistency for normal answer set programs (ASP) is NP-complete, thus, as hard as the satisfaction problem for classical propositional logic (SAT).
The best algorithms to solve these problems take exponential time in the worst case.
The \emph{exponential time hypothesis} (ETH) implies that this result is tight for SAT, that is, SAT cannot be solved in subexponential time.
This immediately establishes that the result is also tight for the consistency problem for ASP.
However, accounting for the treewidth of the problem, the consistency problem for ASP is slightly harder than SAT: while SAT can be solved by an algorithm that runs in exponential time in the treewidth~$k$, it was recently shown that ASP requires exponential time in~$k\cdot\log(k)$.
This extra cost is due checking that there are no self-supported true atoms due to positive cycles in the program.
In this paper, we refine the above result and show that the consistency problem for ASP can be solved in exponential time in~$k\cdot\log(\lambda)$ where~$\lambda$ is the minimum between the treewidth and the size of the largest strongly-connected component in the positive dependency graph of the program.
We provide a dynamic programming algorithm that solves the problem and a treewidth-aware reduction from ASP to SAT that adhere to the above limit.
\end{abstract}

\section{Introduction}\label{sec:introduction}

Answer Set Programming (\ASP)~\cite{BrewkaEiterTruszczynski11,GebserKaminskiKaufmannSchaub12} is a 
problem modeling and solving paradigm well-known in the area 
of knowledge representation and reasoning that is experiencing
an increasing number of successful
applications~\cite{BalducciniGelfondNogueira06a,NiemelaSimonsSoininen99,NogueiraBalducciniGelfond01a,GuziolowskiEtAl13a,SchaubWoltran18}.
The flexibility of \ASP comes with a high computational complexity const:
its \emph{consistency problem}, that is, deciding the existence of
a solution (\emph{answer set}) for a given logic program is~$\Sigma_2^P$-complete~\cite{EiterGottlob95}, in general.
Fragments with lower complexity are also know.
For instance, the consistency problem for \emph{normal \ASP} or \emph{head-cycle-free (HCF) \ASP}, is \NP-complete.
Even for solving this class of programs, the best known algorithms require exponential time with respect to the size of the program.
Still, existing solvers~\cite{gekasc09c,alcadofuleperiveza17a} are able to find solutions for many interesting problems in reasonable time.
A way to shed light into this discrepancy is by means of \emph{parameterized
complexity}~\cite{CyganEtAl15}, which conducts more fine-grained
complexity analysis in terms
of parameters of a problem.
For ASP, several results were achieved in this direction~\cite{GottlobScarcelloSideri02,LoncTruszczynski03,LinZhao04,FichteSzeider15},
some insights involve even combinations~\cite{LacknerPfandler12,FichteKroneggerWoltran19} of parameters.
More recent studies focus on the influence of the parameter
\emph{treewidth} for solving \ASP~\cite{JaklPichlerWoltran09,FichteEtAl17a,FichteHecher19,BichlerMorakWoltran18,BliemEtAl20}.
These works directly make use of the treewidth of a given logic program
in order to solve, e.g., the consistency problem, in polynomial time in the program size,
while being exponential only in the treewidth.
%
%
Recently, it was shown that for normal ASP deciding consistency is
expected to be slightly superexponential for treewidth~\cite{Hecher20}.
More concretely, a lower bound was established saying that under reasonable 
assumptions such as the \emph{Exponential Time Hypothesis (ETH)}~\cite{ImpagliazzoPaturiZane01}, 
consistency for any normal logic program of treewidth~$k$ cannot be decided in 
time significantly better than~$2^{k\cdot\ceil{\log(k)}}\cdot\poly(n)$, 
where~$n$ is the number of variables (atoms) of the program. 
This result matches the known upper bound~\cite{FichteHecher19}
and shows that the consistency of normal \ASP is slightly harder than the satisfiability (\SAT) of a propositional formula, 
which under the ETH cannot be decided in time~$2^{o(k)}\cdot\poly(n)$.

We address this result and provide a more detailed analysis,
where besides treewidth, we also consider the size~$\ell$ of the largest strongly-connected components (SCCs)
of the \emph{positive dependency graph} as parameter.
This allows us to obtain runtimes below~$2^{k\cdot\ceil{\log(k)}}\cdot\poly(n)$ and show that that not all positive cycles of logic programs are equally hard. 
%
%
Then, we also provide a treewidth-aware reduction from head-cycle-free \ASP to the fragment of tight \ASP,
which prohibits cycles in the corresponding positive dependency graph.
This reduction reduces a given head-cycle-free program of treewidth~$k$ to a tight program of treewidth~$\mathcal{O}(k\cdot\log(\ell))$,
which improves known results~\cite{Hecher20}.
Finally, we establish that tight \ASP is as hard as \SAT in terms of treewidth.

\smallskip\noindent
\textbf{Contributions.}
More concretely, we present the following. 
\begin{enumerate}
	\item First, we establish a parameterized algorithm for deciding consistency of any head-cycle-free program~$\Pi$ that runs in time~$2^{\mathcal{O}(k\cdot\log(\ell))}\cdot\poly(\Card{\at(\Pi)})$, where~$k$ is the treewidth of~$\Pi$ and~$\ell$ is the size of the largest strongly-connected component (SCC) of the dependency graph of~$\Pi$.
	Combining this result with results from~\cite{{Hecher20}}, 
	consistency of any head-cycle-free program can be decided in~$2^{\mathcal{O}(k\cdot\log(\lambda))}\cdot\poly(\Card{\at(\Pi)})$
	where~$\lambda$ is the minimum of~$k$ and~$\ell$.
	Besides, our algorithm bijectively preserves answer sets with respect to the atoms of~$\Pi$ and can be therefore easily extended, see, e.g.~\cite{PichlerRuemmeleWoltran10}, 
	for counting and enumerating answer sets. 

	\item Then, we present a treewidth-aware reduction from head-cycle-free ASP to tight ASP.
	Our reduction takes any head-cycle-free  program~$\Pi$ and creates a tight  program,
whose treewidth is at most~$\mathcal{O}(k\cdot\log(\ell))$, where~$k$ is the treewidth of~$\Pi$ and~$\ell$ is the size of the largest SCC of the dependency graph of~$\Pi$.
	%
%
In general, the treewidth of the resulting tight program cannot be in~$o(k\cdot\log(k))$, unless ETH fails.
%
Our reduction forms a major improvement 
for the particular case where~$\ell \ll k$.
%
%
	\item Finally, we show a treewidth-aware reduction that takes any tight logic program~$\Pi$ and
	creates a propositional formula, whose treewidth is linear in the treewidth of the program.
	This reduction cannot be significantly improved under ETH. Our result also establishes that for deciding consistency of tight logic programs of bounded treewidth~$k$, one indeed obtains the same runtime as for \SAT, namley~$2^{\mathcal{O}(k)}\cdot\poly(\Card{\at(\Pi)})$, which is ETH-tight.
\end{enumerate}

\smallskip
\noindent\textbf{Related Work.}
While the largest SCC size has already been considered~\cite{Janhunen06},
it has not been studied in combination with treewidth. 
Also programs, where the number of even and/or odd cycles is bounded, have been analyzed~\cite{LinZhao04}, which is orthogonal to the size of the largest cycle or largest SCC size~$\ell$. Indeed, in the 
worst-case, each component might have an exponential number of cycles in~$\ell$.
%
%
%
%
Further, the literature distinguishes the so-called feedback width~\cite{GottlobScarcelloSideri02}, which involves the number of atoms required to break the positive cycles.
%
%
There are also related measures, called smallest backdoor size, 
where the removal of a backdoor, i.e., set of atoms, from the program results in 
normal or acyclic programs~\cite{FichteSzeider15,FichteSzeider17}.
%
%

\section{Background}\label{sec:background}
We assume familiarity with graph terminology, cf.,~\cite{Diestel12}.
Given a directed graph~$G=(V,E)$. Then, a set~$C\subseteq V$ of vertices of~$G$
is a \emph{strongly-connected component (SCC)} of~$G$ if
$C$ is a~$\subseteq$-largest set such that 
for every two distinct vertices~$u,v$ in~$C$ there is a directed path from~$u$ to~$v$ in~$G$.
A \emph{cycle} over some vertex $v$ of~$G$ is a directed path from~$v$ to~$v$.

\smallskip
\noindent\textbf{Answer Set Programming (ASP).} %
%
We assume familiarity with propositional satisfiability (\SAT)~\cite{BiereHeuleMaarenWalsh09,KleineBuningLettman99}, 
and follow standard definitions of propositional ASP~\cite{BrewkaEiterTruszczynski11,JanhunenNiemela16a}.
%
Let $m$, $n$, $o$ be non-negative integers such that
$m \leq n \leq o$, $a_1$, $\ldots$, $a_o$ be distinct propositional
atoms. Moreover, we refer by \emph{literal} to an atom or the negation
thereof.
%
A \emph{(logic) program}~$\Pi$ is a set of \emph{rules} of the form
%
\(
a_1\por \cdots \por a_m \hsep a_{m+1}, \ldots, a_{n}, \neg
a_{n+1}, \ldots, \neg a_o.
\)
%
%
%
%
%
%
%
%
%
%
For a rule~$r$, we let $H_r \eqdef \{a_1, \ldots, a_m\}$,
$B^+_r \eqdef \{a_{m+1}, \ldots, a_{n}\}$, and
$B^-_r \eqdef \{a_{n+1}, \ldots, a_o\}$.
%
%
%
We denote the sets of \emph{atoms} occurring in a rule~$r$ or in a
program~$\Pi$ by $\at(r) \eqdef H_r \cup B^+_r \cup B^-_r$ and
$\at(\Pi)\eqdef \bigcup_{r\in\Pi} \at(r)$.
For a set~$X\subseteq\at(\Pi)$ of atoms, we let~$\overline{X}\eqdef \{\neg x\mid x\in X\}$.
%
%
%
Program~$\Pi$ is \emph{normal}, if $\Card{H_r} \leq 1$ for
every~$r \in \Pi$.
The \emph{positive dependency digraph}~$D_\Pi$ of $\Pi$ is the
directed graph defined on the set of atoms
from~$\bigcup_{r\in \Pi}H_r \cup B^+_r$, where there is a directed edge from vertex~$a$ to vertex~$b$ iff there is a rule~$r \in \Pi$ with $a\in B^+_r$ and~$b\in H_r$.
A head-cycle of~$D_\Pi$ is an $\{a, b\}$-cycle\footnote{Let
  $G=(V,E)$ be a digraph and $W \subseteq V$. Then, a cycle in~$G$ is
  a $W$-cycle if it contains all vertices from~$W$.} for two distinct
atoms~$a$, $b \in H_r$ for some rule $r \in \Pi$. 
A program~$\Pi$ is 
\emph{head-cycle-free (HCF)} if $D_\Pi$ contains no
head-cycle~\cite{Ben-EliyahuDechter94} and~$\Pi$
is called \emph{tight} if~$D_\Pi$ contains no cycle at all~\cite{LinZhao03}.
The class of tight, normal, and HCF programs is referred to by
\emph{tight, normal, and HCF ASP}, respectively.
%
%

An \emph{interpretation} $I$ is a set of atoms. $I$ \emph{satisfies} a
rule~$r$ if $(H_r\,\cup\, B^-_r) \,\cap\, I \neq \emptyset$ or
$B^+_r \setminus I \neq \emptyset$.  $I$ is a \emph{model} of $\Pi$
if it satisfies all rules of~$\Pi$, in symbols $I \models \Pi$. 
For brevity, we view propositional formulas 
as sets of clauses that need to be satisfied, and
use the notion of interpretations, models, and satisfiability analogously. 
%
The \emph{Gelfond-Lifschitz
  (GL) reduct} of~$\Pi$ under~$I$ is the program~$\Pi^I$ obtained
from $\Pi$ by first removing all rules~$r$ with
$B^-_r\cap I\neq \emptyset$ and then removing all~$\neg z$ where
$z \in B^-_r$ from every remaining
rule~$r$~\cite{GelfondLifschitz91}. %
$I$ is an \emph{answer set} of a program~$\Pi$ if $I$ is a minimal
model of~$\Pi^I$. %
%
The problem of deciding whether an \ASP program has an answer set is called
\emph{consistency}, which is \SIGMA{2}{P}-complete~\cite{EiterGottlob95}. 
If the input is restricted to normal programs, the complexity drops to
\NP-complete~\cite{BidoitFroidevaux91,MarekTruszczynski91}.
A head-cycle-free program~$\Pi$ 
can be translated into a normal program in polynomial
time~\cite{Ben-EliyahuDechter94}.
%
%
%
%
%
%
%
%
%
%
%
%
%
%
%
The following characterization of answer sets is often
invoked when considering normal programs~\cite{LinZhao03}.
Given a set~$A\subseteq\at(\Pi)$ of atoms, a function~$\sigma: A \rightarrow \{0,\ldots,\Card{A}-1\}$ is called \emph{level mapping} over~$A$.
Given a model~$I$ of a normal program~$\Pi$ and a level mapping~$\sigma$ over~$I$, an atom~$a\in I$ is \emph{proven} 
if there is a rule~$r\in\Pi$ \emph{proving~$a$ with~$\sigma$}, where $a\in H_r$ with (i)~$B^+_r\subseteq I$,
(ii)~$I \cap B^-_r = \emptyset$ and
$I \cap (H_r \setminus \{a\}) = \emptyset$,
and (iii)~$\sigma(b) < \sigma(a)$ for every~$b\in B_r^+$. Then, $I$ is an
\emph{answer set} of~$\Pi$ if (i)~$I$ is a model of~$\Pi$, and
(ii) \emph{$I$ is proven}, i.e., every~$a \in I$ is proven.
This characterization vacuously extends to head-cycle-free
programs~\cite{Ben-EliyahuDechter94} and
allows for further simplification when considering SCCs of~$D_\Pi$~\cite{Janhunen06}.
To this end, we denote for each atom~$a\in\at(\Pi)$ the \emph{strongly-connected component (SCC)} of atom~$a$ in~$D_\Pi$ by~$\scc(a)$.
Then, Condition (iii) above can be relaxed to~$\sigma(b) < \sigma(a)$
for every~$b\in B_r^+\cap C$, where~$C=\scc(a)$ is the SCC of~$a$.

\begin{figure}[t]%
  \centering
  \shortversion{ %
    \begin{tikzpicture}[node distance=7mm,every node/.style={fill,circle,inner sep=2pt}]
\node (a) [label={[text height=1.5ex,yshift=-0.075cm,xshift=0.00cm]above:$b$}] {};
\node (b) [right of=a,label={[text height=1.5ex,yshift=0.09cm,xshift=-0.07cm]right:$a$}] {};
\node (f) [left of=a,label={[text height=.85ex,xshift=0.25em]left:$e$}] {};
\node (e) [below of=a, label={[text height=1.5ex,xshift=-.0em,yshift=.08cm]below:$c$}] {};
\node (d) [below of=b,label={[text height=1.5ex,yshift=0.09cm,xshift=-0.07cm]right:$d$}] {};
\draw [<-,thick] (a) to (b);
\draw [<-,thick] (b) to (d);
\draw [<-,thick] (d) to (e);
\draw [<-,thick] (e) to (a);
\draw [<->,thick] (d) to (a);
\draw [<-,thick] (a) to (f);
\end{tikzpicture}%
    \vspace{-.4em}
    \caption{Positive dependency graph~$D_\Pi$ of program~$\Pi$ of Example~\ref{ex:running1}.}
  }%
  \label{fig:depgraph}
\end{figure}
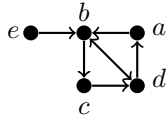

\begin{example}
\label{ex:running1}\label{ex:running}
Consider the following program
\vspace{-0.1em}
$\Pi\eqdef$
%
$\SB\overbrace{a \hsep d}^{r_1};\, %
\overbrace{b \hsep a}^{r_2};\, %
\overbrace{b \hsep d}^{r_3};$ 
$\overbrace{b \hsep e, \neg f}^{r_4};\, %
\overbrace{c \hsep b}^{r_5};\, %
\overbrace{d \hsep b, c}^{r_6};\, %
\overbrace{e \lor f \lor g \hsep}^{r_7} %
\SE$.
%
%
Observe that $\Pi$ is head-cycle-free.
Figure~\ref{fig:depgraph} shows the positive dependency graph~$D_\Pi$
consisting of SCCs~$\scc(e)$ and~$\scc(a)=\scc(b)=\scc(c)=\scc(d)$.
Then, $I\eqdef\{a,b, c, d,e\}$ is an answer set of~$\Pi$,
since~$I\models\Pi$, and we can prove with level mapping
${\sigma \eqdef\{e\mapsto 0}, {f\mapsto 0},\allowbreak {g\mapsto 0},{ b\mapsto 0}, {c\mapsto 1}, {d\mapsto 2}, {a\mapsto 3\}}$
atom~$e$ by rule~$r_7$, 
atom~$b$ by rule~$r_4$,
atom~$c$ by rule~$r_5$, and
atom~$d$ by rule~$r_6$.
Further answer sets are $\{f\}$ and
$\{g\}$.

\end{example}%

\smallskip
\noindent\textbf{Tree Decompositions (TDs).} %
%
%
%
A \emph{tree decomposition (TD)}~\cite{RobertsonSeymour86} 
of a given graph~$G{=}(V,E)$ is a pair
$\TTT{=}(T,\chi)$ where $T$ is a tree rooted at~$\rootOf(T)$ and $\chi$ 
assigns to each node $t$ of~$T$ a set~$\chi(t)\subseteq V$,
called \emph{bag}, such that (i) $V=\bigcup_{t\text{ of }T}\chi(t)$, (ii)
$E\subseteq\SB \{u,v\} \SM t\text{ in } T, \{u,v\}\subseteq \chi(t)\SE$,
and (iii) ``connectedness'': for each $r, s, t\text{ of } T$, such that $s$ lies on the path
from~$r$ to $t$, we have $\chi(r) \cap \chi(t) \subseteq \chi(s)$.
For every node~$t$ of~$T$, we denote by $\children(t)$ the \emph{set of child nodes of~$t$} in~$T$.
The \emph{bags~$\chi_{\leq t}$ below~$t$} consists of the union of all bags of nodes below~$t$ in~$T$, including~$t$. 
We
let $\width(\TTT) {\eqdef} \max_{t\text{ of } T}\Card{\chi(t)}-1$.
The
\emph{treewidth} $\tw{G}$ of $G$ is the minimum $\width({\TTT})$ over
all TDs $\TTT$ of $G$. TDs can be 5-approximated in \emph{single exponential time}~\cite{BodlaenderEtAl13} in the treewidth. 
For a node~$t \text{ of } T$, we say that $\type(t)$ is $\leaf$ if $t$ has
no children and~$\chi(t)=\emptyset$; $\join$ if $t$ has children~$t'$ and $t''$ with
$t'\neq t''$ and $\chi(t) = \chi(t') = \chi(t'')$; $\intr$
(``introduce'') if $t$ has a single child~$t'$,
$\chi(t') \subseteq \chi(t)$ and $|\chi(t)| = |\chi(t')| + 1$; $\rem$
if $t$ has a single child~$t'$,
$\chi(t') \supseteq \chi(t)$ and $|\chi(t')| = |\chi(t)| + 1$. If for
every node $t\text{ of } T$, %
$\type(t) \in \{ \leaf, \join, \intr, \rem\}$, the TD is called \emph{nice}.
%
A TD can be turned into a nice TD~\cite{Kloks94a}[Lem.\ 13.1.3] \emph{without increasing the width} in linear~time.
%

\begin{example}\label{ex:td}
  Figure~\ref{fig:graph-td} illustrates a graph~$G$ and a TD~$\mathcal{T}$ of~$G$ of width~$2$, which is also the treewidth of~$G$,
  since~$G$ contains~\cite{Kloks94a} a completely connected graph among vertices $b$,$c$,$d$.
\end{example}

In order to use TDs for \ASP, we need
dedicated graph representations of programs~\cite{JaklPichlerWoltran09}.
The \emph{primal graph\footnote{Analogously, the primal graph~$G_F$ of a propositional Formula~$F$ %
uses variables of~$F$ as vertices and adjoins two vertices~$a,b$ by an edge, if there is a formula in~$F$ %
containing~$a,b$.%
}}~$G_\Pi$
of program~$\Pi$ has the atoms of~$\Pi$ as vertices and an
edge~$\{a,b\}$ if there exists a rule~$r \in \Pi$ and $a,b \in \at(r)$.
\longversion{The \emph{incidence graph}~$I_\Pi$ of $\Pi$ is the bipartite
graph that has the atoms and rules of~$\Pi$ as vertices and an
edge~$a\, r$ if $a \in \at(r)$ for some rule~$r \in \Pi$.}
\noindent Let ${\cal T} = (T, \chi)$ be a TD of primal
graph~$G_\Pi$ of a program $\Pi$, and let~$t$ be a node of~$T$. 
%
The \emph{bag program}~$\Pi_t$ contains rules entirely covered by
the bag~$\chi(t)$. Formally, $\Pi_t \eqdef \SB r \SM r \in \Pi, \at(r) \subseteq \chi(t)\}$.
\begin{example}
  Recall program~$\Pi$ from Example~\ref{ex:running1} and observe
  that graph~$G$ of Figure~\ref{fig:graph-td} is the primal graph
  of~$\Pi$.
  Further, we have $\Pi_{t_1}=\{r_1,r_2,r_3\}$, $\Pi_{t_2}=\{r_3,r_5,r_6\}$,
  $\Pi_{t_3}=\emptyset$, $\Pi_{t_4}=\{r_7\}$ and~$\Pi_{t_5}=\{r_4\}$. 
  \longversion{%
    Further, graph~$G_2$ of Figure~\ref{fig:graph-td2} is the
    incidence graph of~$\Pi$.
  }
\end{example}

%

\begin{figure}[t]%
  \centering
  \shortversion{ %
    \begin{tikzpicture}[node distance=7mm,every node/.style={fill,circle,inner sep=2pt}]
\node (a) [label={[text height=1.5ex,yshift=-0.075cm,xshift=0.00cm]above:$b$}] {};
\node (b) [right of=a,label={[text height=1.5ex,yshift=0.09cm,xshift=-0.07cm]right:$a$}] {};
\node (f) [left of=a,label={[text height=.85ex,xshift=0.25em]left:$e$}] {};
\node (e) [below of=a, label={[text height=1.5ex,xshift=-.0em,yshift=.08cm]below:$c$}] {};
\node (d) [below of=b,label={[text height=1.5ex,yshift=0.09cm,xshift=-0.07cm]right:$d$}] {};
\node (c) [left of=e,label={[text height=1.5ex,yshift=0.09cm,xshift=0.16cm]left:$f$}] {};
\draw (a) to (c);
\draw (a) to (b);
\draw (b) to (d);
\draw (d) to (e);
\draw (e) to (a);
\draw (d) to (a);
\draw (c) to (f);
\draw (a) to (f);
\end{tikzpicture}%
    \includegraphics{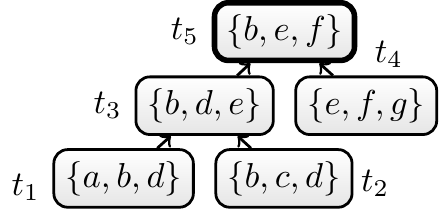}
    \vspace{-.4em}
    \caption{Graph~$G$ (left) and a tree decomposition~$\mathcal{T}$ of~$G$ (right).}
  }%
  \longversion{%
    \begin{subfigure}[c]{0.47\textwidth}
      \centering%
      \begin{tikzpicture}[node distance=7mm,every node/.style={fill,circle,inner sep=2pt}]
\node (a) [label={[text height=1.5ex,yshift=-0.075cm,xshift=0.00cm]above:$b$}] {};
\node (b) [right of=a,label={[text height=1.5ex,yshift=0.09cm,xshift=-0.07cm]right:$a$}] {};
\node (f) [left of=a,label={[text height=.85ex,xshift=0.25em]left:$e$}] {};
\node (e) [below of=a, label={[text height=1.5ex,xshift=-.0em,yshift=.08cm]below:$c$}] {};
\node (d) [below of=b,label={[text height=1.5ex,yshift=0.09cm,xshift=-0.07cm]right:$d$}] {};
\node (c) [left of=e,label={[text height=1.5ex,yshift=0.09cm,xshift=0.16cm]left:$f$}] {};
\draw (a) to (c);
\draw (a) to (b);
\draw (b) to (d);
\draw (d) to (e);
\draw (e) to (a);
\draw (d) to (a);
\draw (c) to (f);
\draw (a) to (f);
\end{tikzpicture}%
      \input{graph0/td}%
      \caption{Graph~$G_1$ and a tree decomposition of~$G_1$.}
      \label{fig:graph-td}
    \end{subfigure}
    \begin{subfigure}[c]{0.5\textwidth}
      \centering \input{graph0/graph_inc}%
      \input{graph0/td_inc}%
      \caption{Graph~$G_2$ and a tree decomposition of~$G_2$.}
      \label{fig:graph-td2}%
    \end{subfqigure}
    \caption{Graphs~$G_1, G_2$ and two corresponding tree
      decompositions.}
  }%
  \label{fig:graph-td}%
\end{figure}

\section{Bounding Treewidth and Positive Cycles}\label{sec:approach}

Recently, it was shown that under reasonable assumptions, namely
the exponential time hypothesis (ETH), deciding consistency of normal logic programs is slightly superexponential
and one cannot expect to significantly improve in the worst case.
For a given normal logic program, where~$k$ is the treewidth of the primal graph of the program,
this implies that one cannot decide consistency in time significantly better than~$2^{k\cdot\ceil{\log(k)}}\cdot\poly(\Card{\at(\Pi)})$. 
%
%
\begin{proposition}[Lower Bound for Treewidth~\cite{Hecher20}]\label{prop:lowerbound}
Given a normal or head-cycle-free logic program~$\Pi$, where~$k$ is the treewidth of the primal graph of~$\Pi$.
Then, under ETH one cannot decide consistency of~$\Pi$ in time~$2^{o(k\cdot\log(k))}\cdot\poly(\Card{\at(\Pi)})$.
\end{proposition}

While according to Proposition~\ref{prop:lowerbound}, we cannot expect to significantly improve the runtime for normal logic programs
in the worst case, it still is worth to study the underlying reason that makes the worst case so bad.
It is well-known that positive cycles are responsible for the hardness~\cite{LifschitzRazborov06,Janhunen06} of computing
answer sets of normal logic programs.
The particular issue with logic programs~$\Pi$ in combination with treewidth and large cycles is that in a tree decomposition of~$G_\Pi$
it might be the case that the cycle spreads across the whole decomposition,
i.e., tree decomposition bags only contain parts of such cycles, which prohibits to view these cycles (and dependencies) as a whole.
This is also the reason of the hardness given in Proposition~\ref{prop:lowerbound} and explains why under bounded treewidth 
evaluating normal logic programs is harder than evaluating proposition formulas.
However, if a given normal logic program only has positive cycles of lengths at most~$3$,
and each atom appears in at most one positive cycle, the properties of tree decompositions
already ensure that the atoms of each such positive cycle appear in at least one common bag.
Indeed, a cycle of length at most~$3$ forms a completely connected subgraph
and therefore it is guaranteed~\cite{Kloks94a} that the atoms of the cycle are in one common bag of any tree decomposition of~$G_\Pi$.
\begin{example}
Recall program~$\Pi$ of Example~\ref{ex:running1}.
Observe that in any TD of~$G_\Pi$ it is required that there are nodes~$t,t'$ with~$\chi(t)\subseteq \{b,c,d\}$ and~$\chi(t')\subseteq \{a,b,d\}$
since a cycle of length~$3$ in the positive dependency graph~$D_\Pi$ (cf., Figure~\ref{fig:depgraph}) forms a completely connected graph in the primal graph, cf., Figure~\ref{fig:graph-td} (left). 
\end{example}

In the following, we generalize this result to cycles of length at most~$\ell$, where we bound the size of these positive cycles 
in order to improve the lower bound of Proposition~\ref{prop:lowerbound}
on programs of bounded positive cycle lengths.
This will provide not only a significant improvement in the running time on programs, where the size of positive cycles is bounded,
but also shows that indeed the case of positive cycle lengths up to~$3$ can be generalized to lengths beyond~$3$.
Consequently, we establish that not all positive cycles are bad assuming that the maximum size~$\ell$ of the positive cycles is bounded,
which provides an improvement of Proposition~\ref{prop:lowerbound} as long as~$\ell \ll k$, where~$k$ is the treewidth of~$G_\Pi$.

\subsubsection*{Bounding positive Cycles.}
In the remainder, we assume a HCF logic program~$\Pi$, whose treewidth is given by~$k=\tw{G_\Pi}$. 
We let~$\ell_{\scc(a)}$ for each atom~$a$ be the \emph{number of atoms (size) of the SCC} of~$a$ in~$D_\Pi$.
Further, we let~$\ell\eqdef \max_{a\in\at(\Pi)}\ell_{\scc(a)}{+1}$ be the \emph{largest SCC size}.
This also bounds the lengths of positive cycles. If each atom~$a$ appears in at most one positive cycle, we have that~$\ell_{\scc(a)}$ 
is the cycle length of~$a$ and then~$\ell{-}1$ is the length of the largest cycle in~$\Pi$.
We refer to the class of HCF logic programs, whose largest SCC size is bounded by a parameter~$\ell$
by \emph{SCC-bounded \ASP}.
Observe that the largest SCC size~$\ell$ is orthogonal to the measure treewidth.
\begin{example}
Consider program~$\Pi$ from Example~\ref{ex:running1}. Then,~$\ell_{\scc(e)}=1$, $\ell_{\scc(a)}=\ell_{\scc(b)}=\ell_{\scc(c)}=\ell_{\scc(d)}=4$, and~$\ell=5$.

Now, assume a program, whose primal graph equals the dependency graph, which is just one large (positive) cycle. It is easy to see that this program has treewidth~$2$ and
one can define a TD of~$G_\Pi$, whose bags are constructed along the cycle.
However, the largest SCC size coincides with the number of atoms. 
Conversely, there are instances of large treewidth without any positive cycle. 
\end{example}

Bounding cycle lengths or sizes of SCCs seems similar to the non-parameterized context, 
where the consistency of normal logic programs is compiled
to a propositional formula (\SAT) by a reduction based on level mappings that is applied on a
SCC-by-SCC basis~\cite{Janhunen06}.
However, this reduction does not preserve the treewidth.
On the other hand, while our approach also uses level mappings
and proceeds on an SCC-by-SCC basis, the overall evaluation
is not SCC-based, since this might completely destroy the treewidth in the worst-case.
Instead, the evaluation is still guided along a tree decomposition, which is presented
in two flavors. First, we show a dedicated parameterized algorithm for the evaluation of logic programs
of bounded treewidth, followed by a treewidth-aware reduction to propositional satisfiability.

\subsection{An Algorithm for SCC-bounded \ASP and Treewidth}

In the course of this section, we establish the following theorem.

\begin{theorem}[Runtime of SCC-bounded \ASP]\label{theo:runtimescc}
Assume a HCF logic program~$\Pi$,
where the treewidth of the primal graph~$G_\Pi$ of~$\Pi$ is at most~$k$
and~$\ell$ is the largest SCC size.
Then, there is an algorithm for deciding the consistency of~$\Pi$,
running in time~$2^{\mathcal{O}(k\cdot \log(\lambda))} \cdot \poly(\Card{\at(\Pi)})$,
where~$\lambda=\min(\{k,\ell\})$. 
\end{theorem}

The overall idea of the algorithm relies on so-called dynamic programming,
which be briefly recap next.

\subsubsection*{Dynamic Programming on Tree Decompositions.}
\emph{Dynamic programming (DP)} on TDs, see, e.g.,~\cite{BodlaenderKoster08}, 
evaluates a given input instance~$\mathcal{I}$ in parts along a given TD of a graph representation~$G$ of the instance.
Thereby, for each node~$t$ of the TD, intermediate results are 
stored in a \emph{table}~$\tab{t}$. %
This is achieved by running a \emph{table algorithm}, 
which is designed for a certain graph representation, 
and stores in~$\tab{t}$ results of problem parts of~$\mathcal{I}$,
thereby considering tables~$\tab{t'}$ for child nodes~$t'$ of~$t$. 
DP works for \emph{many problems} as follows. 
\vspace{-.2em}\begin{enumerate}
\item Construct a \emph{graph representation}~$G$ of~$\mathcal{I}$.\vspace{-.2em}
\item Compute a TD~$\TTT=(T,\chi)$ of~$G$. For simplicity and better presentation of the different cases within our table algorithms, we use \emph{nice} TDs for DP.
\vspace{-.2em}
\item\label{step:dp} Traverse the nodes of~$T$ in
  post-order (bottom-up tree traversal of~$T$).
  At every node~$t$ of $T$ during post-order traversal, execute a table algorithm
  that takes as input a bag $\chi(t)$, a certain \emph{bag instance}~$\mathcal{I}_t$ depending on the problem, as well as previously computed child tables of~$t$. Then, the results of this execution is stored in table~$\tab{t}$.\vspace{-.2em}
\item Finally, interpret table~$\tab{n}$ for the root node~$n$ of~$T$ in order to \emph{output the solution} to the problem for instance~$\mathcal{I}$.
\end{enumerate}

%

Now, the missing ingredient for solving problems via dynamic programming along a given TD, is a suitable table algorithm.
Such algorithms have been already presented for~\SAT~\cite{SamerSzeider10b} and \ASP~\cite{JaklPichlerWoltran09,FichteEtAl17a,FichteHecher19}. 
We only briefly sketch the ideas of a table algorithm using the primal graph that computes models 
of a given program~$\Pi$. 
%
Each table~$\tab{t}$ consist of rows storing interpretations
over atoms in the bag~$\chi(t)$.
Then, the table~$\tab{t}$ for leaf nodes~$t$ 
consist of the empty interpretation.
For nodes~$t$ with introduced variable $a\in\chi(t)$, 
we store in~$\tab{t}$ 
interpretations of the child table, but 
for each such interpretation 
we decide whether~$a$ is in the interpretation or not, 
and ensure that the interpretation satisfies $\Pi_t$.
When an atom~$b$ is forgotten in a forget node~$t$, we store interpretations of the child table, but restricted to atoms in~$\chi(t)$.
By the properties of a TD, it is then guaranteed that all rules containing~$b$ 
have been processed so far.
For join nodes, we store in~$\tab{t}$ interpretations 
that are also in both child tables of~$t$.

%

\renewcommand{\eqdef}{\leftarrow}
%
%
 \begin{algorithm}[t]
   \KwData{Node~$t$, bag $\chi(t)$, bag program~$\Pi_t$,
     sequence $\langle \tab{1}, \ldots,\tau_o \rangle$ of child tables
     of~$t$.}
  \KwResult{Table~$\tab{t}$.}
   \lIf(\hspace{-1em})
   {$\type(t) = \leaf$}{
     $\tab{t} \eqdef \{ \langle
     \tuplecolor{\inputPredColor}{\emptyset}, \tuplecolor{\outputPredColor}{\emptyset}, \tuplecolor{\statePredColor}{\emptyset}
     \rangle \}$\label{line:primleaf}%
     %
   }%
  \uElseIf{$\type(t) = \intr$ and $a\hspace{-0.1em}\in\hspace{-0.1em}\chi(t)$ is the introduced atom}{
   \vspace{-0.05em}
   \makebox[2.8cm][l]{\hspace{-1em}$\tab{t} \eqdef 
   \{ \langle \tuplecolor{\inputPredColor}{I'}, \tuplecolor{\outputPredColor}{\mathcal{P}'}, \tuplecolor{\statePredColor}{\sigma'} \rangle$}$\mid\langle \tuplecolor{\inputPredColor}{I}, \tuplecolor{\outputPredColor}{\mathcal{P}}, \tuplecolor{\statePredColor}{\sigma} \rangle\in \tab{1},$ $I' \in \{I, \MAI{I}{a}\}, I' \models \Pi_t, $\label{line:primintr1} 

     \makebox[1.4cm][l]{} 
      $ \sigma' \in \possord(\sigma, \{a\} \cap I'), \isminimal(\sigma', \Pi_t), \mathcal{P}'=\mathcal{P} \cup \gatherproof(I', \sigma', \Pi_t)\}
      \hspace{-5em}$\label{line:primintr2}
   \vspace{-0.05em}
     }\vspace{-0.05em}%
     \uElseIf{$\type(t) = \rem$ and $a \not\in \chi(t)$ is the forgotten atom}{
       \makebox[2.8cm][l]{\hspace{-1em}$\tab{t} \eqdef 
       \{ \langle \tuplecolor{\inputPredColor}{\MAR{I}{a}}, \tuplecolor{\outputPredColor}{\MAR{\mathcal{P}}{a}}, \tuplecolor{\statePredColor}{\MARR{\sigma}{a}}
       \rangle$}$|\;\langle \tuplecolor{\inputPredColor}{I}, \tuplecolor{\outputPredColor}{\mathcal{P}}, \tuplecolor{\statePredColor}{\sigma}
       \rangle \in \tab{1}, a \in \mathcal{P} \cup (\{a\} \setminus I) \}
       \hspace{-5em}$\label{line:primrem}
       \vspace{-0.1em}
     } %
     \uElseIf{$\type(t) = \join\qquad\tcc*[h]{$o{=}2$ children of~$t$}$}{
       \makebox[2.8cm][l]{\hspace{-1em}$\tab{t} \eqdef 
       \{ \langle \tuplecolor{\inputPredColor}{I}, \tuplecolor{\outputPredColor}{\mathcal{P}_1 \cup \mathcal{P}_2}, \tuplecolor{\statePredColor}{\sigma}
         \rangle$}$|\;\langle \tuplecolor{\inputPredColor}{I}, \tuplecolor{\outputPredColor}{\mathcal{P}_1}, \tuplecolor{\statePredColor}{\sigma} \rangle \in \tab{1}, \langle \tuplecolor{\inputPredColor}{I}, \tuplecolor{\outputPredColor}{\mathcal{P}_2}, \tuplecolor{\statePredColor}{\sigma} \rangle \in \tab{2}\}\hspace{-5em}$\label{line:primjoin}
       \vspace{-0.1em}
     } 
     \Return $\tab{t}$
     \vspace{-0.25em}
     \caption{Table algorithm~$\PRIM(t, \chi(t), \Pi_t,
       \langle \tau_1, \ldots, \tau_o \rangle)$ for nodes of nice TDs.}
 \label{fig:prim}\algorithmfootnote{
\renewcommand{\eqdef}{{\ensuremath{\,\mathrel{\mathop:}=}}}
  For a function~$\sigma$ mapping~$x$ to~$\sigma(x)$, we let $\MARR{\sigma}{x} \eqdef \sigma\setminus \{x\mapsto \sigma(x)\}$ be the function~$\sigma$ without containing~$x$.
Further, for given set~$S$ and an element~$e$, we let
  $\MAI{S}{e} \eqdef S \cup \{e\}$ and
  $\MAR{S}{e} \eqdef S \setminus \{e\}$.}
\end{algorithm}%
\renewcommand{\eqdef}{{\ensuremath{\,\mathrel{\mathop:}=}}}

\subsection{Exploiting Treewidth for SCC-bounded \ASP}

Similar to the table algorithm sketched above, we present next a table algorithm~\PRIM for solving 
consistency of SCC-bounded \ASP.
Let therefore~$\Pi$ be a given SCC-bounded program of largest SCC size~$\ell$
and~$\mathcal{T}=(T,\chi)$ be a tree decomposition of~$G_\Pi$.
Before we discuss the tables and the algorithm itself,
we need to define level mappings similar to related work~\cite{Janhunen06}, but adapted
to SCC-bounded programs.
Formally, a \emph{level mapping}~$\sigma: A\rightarrow \{0,\ldots,\ell{-}2\}$ over atoms~$A\subseteq\at(\Pi)$ is a function mapping each atom~$a\in A$ 
to a level~$\sigma(a)$ such that the level does not exceed~$\ell_{\scc(a)}$, i.e.,~$\sigma(a)<\ell_{\scc(a)}$.
These level mappings are used in the construction of the tables of~\PRIM,
where each table~$\tab{t}$ for a node~$t$ of TD~$\mathcal{T}$ 
consists of rows of the form $\langle I, \mathcal{P}, \sigma\rangle$,
where~$I\subseteq \chi(t)$ is an interpretation of atoms~$\chi(t)$, $\mathcal{P}\subseteq\chi(t)$
is a set of atoms in~$\chi(t)$ that are proven, and~$\sigma$ is a level mapping over~$\chi(t)$.
Before we discuss the table algorithm, we need auxiliary notation.
Let~$\gatherproof(I, \sigma, \Pi_t)$ be a subset of atoms~$I$
containing all atoms~$a\in I$
where there is a rule~$r\in\Pi_t$ proving~$a$ with~$\sigma$.
However, $\sigma$ provides for~$a$ only level numbers \emph{within} the SCC of~$a$,
i.e.,~$\gatherproof$ 
requires the relaxed characterization of 
provability that considers~$\scc(a)$, as given in Section~\ref{sec:background}.
%
%
Then, we denote by~$\possord(\sigma, I)$ those
set of level mappings~$\sigma'$ that extend~$\sigma$ by atoms in~$I$,
where for each atom~$a\in I$, we have a level~$\sigma'(a)$ with~$\sigma'(a)< \ell_{\scc(a)}$.
Further, we let~$\isminimal(\sigma,\Pi_t)$ be $0$ if~$\sigma$ is not \emph{minimal},
i.e., if there is an atom~$a$ with~$\sigma(a)>0$ where
a rule~$r\in\Pi_t$ proves~$a$ with a level mapping~$\rho$ that is identical to~$\sigma$, but sets~$\rho(a)=\sigma(a)-1$, and be~$1$ otherwise.
Listing~\ref{fig:prim} depicts an algorithm~\PRIM for solving consistency of SCC-bounded \ASP.
The algorithm is inspired by an approach for HCF logic programs~\cite{FichteHecher19}, whose idea 
is to evaluate~$\Pi$ in parts, given by the
tree decomposition~$\mathcal{T}$. 
For the ease of presentation, algorithm~\PRIM is presented for nice tree decompositions,
where we have a clear case distinction for every node~$t$ depending on
the node type~$\type(t)\in\{\leaf, \intr, \rem, \join\}$.
For arbitrary decompositions the cases are interleaved.
If~$\type(t)=\leaf$, we have that~$\chi(t)=\emptyset$ and therefore
for~$\chi(t)$ the interpretation, the set of proven atoms as well as the level mapping is empty, cf. Line~\ref{line:primleaf} of Listing~\ref{fig:prim}.
Whenever an atom~$a\in\chi(t)$ is introduced, i.e., if $\type(t)=\intr$,
we construct succeeding rows of the form~$\langle I', \mathcal{P}', \sigma'\rangle$ for every
row in the table~$\tab{1}$ of the child node of~$t$.
We take such a row~$\langle I, \mathcal{P}, \sigma\rangle$ of~$\tab{1}$ and guess whether~$a$ is in~$I$, resulting in~$I'$, and ensure that~$I'$ satisfies~$\Pi_t$,
as given in Line~\ref{line:primintr1}.
Then, Line~\ref{line:primintr2} takes succeeding level mappings~$\sigma'$ of~$\sigma$, as given by~$\possord$, that are minimal (see~$\isminimal$) and we finally ensure that
the proven atoms~$\mathcal{P}'$ update~$\mathcal{P}$ by~$\gatherproof(I',\sigma',\Pi_t)$.
Notably, if duplicate answer sets are not an issue, one can remove the occurence
of~$\isminimal$ in Line~\ref{line:primintr2}.
Whenever an atom~$a$ is forgotton in node~$t$, i.e., if $\type(t)=\rem$,
we take in Line~\ref{line:primrem} only rows of the table~$\tab{1}$ for the child node of~$t$,
where either~$a$ is not in the interpretation or~$a$ is proven,
and remove~$a$ from the row accordingly.
By the properties of TDs, it is guaranteed that we have encountered all rules
involving~$a$ in any node below~$t$.
Finally, if~$t$ is a join node ($\type(t)=\join$), we ensure in Line~\ref{line:primjoin} that
we take only rows of both child tables of~$t$, which agree on interpretations and
level mappings, and that an atom is proven if it is proven in one of the two child rows.

\begin{figure}[t]%
  \centering
  \shortversion{ %
    \includegraphics[scale=0.9]{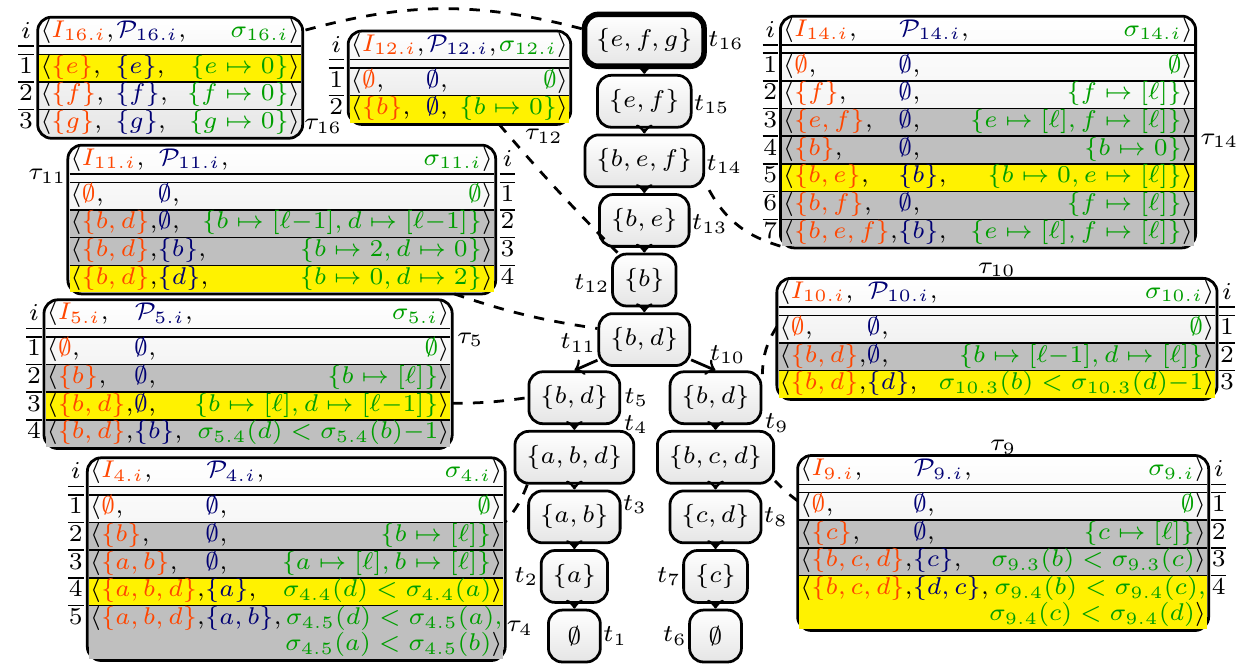}
    \vspace{-1em}
    \caption{Tables obtained by DP on a TD~$\mathcal{T}'$ using algorithm~\PRIM of Listing~\ref{fig:prim}.}\label{fig:tables}
  }%
\end{figure}
\begin{example}
Recall program~$\Pi$ with~$\ell=5$ from Example~\ref{ex:running1}.
Figure~\ref{fig:tables} shows a nice TD~$\mathcal{T}'$ of~$G_\Pi$
and lists selected tables~$\tab{1}, \ldots, \tab{16}$ that
are obtained during DP by using~$\PRIM$ (cf., Listing~\ref{fig:prim})
on TD~$\mathcal{T}'$. Rows highlighted in gray are discarded and do not lead to an answer set, yellow highlighted rows form one answer set.
For brevity, we compactly represent tables by grouping rows according to similar level mappings.
We write~$[\ell]$ for any value in~$\{0,\ldots,\ell{-}2\}$ and 
we sloppily write, e.g.,~$\sigma_{9.3}(b)<\sigma_{9.3}(c)$ to indicate any level mapping~$\sigma_{9.3}$ in row~$3$ of table~$\tab{9}$, where~$b$ has a smaller level than~$c$.

Node~$t_1$ is a leaf ($\type(t_1)=\leaf$) and therefore~$\tab{1}=\{\langle\emptyset,\emptyset,\emptyset\rangle\}$ as stated in Line~\ref{line:primleaf}.
Then, nodes~$t_2,t_3$ and~$t_4$ are introduce nodes.
Therefore, table~$\tab{4}$ is the result of Lines~\ref{line:primintr1} and~\ref{line:primintr2} executed for nodes~$t_2,t_3$ and~$t_4$, by introducing $a,b$, and~$d$, respectively.
Table~$\tab{4}$ contains all interpretations restricted to~$\{a,b,d\}$
that satisfy~$\Pi_{t_4}=\{r_1,r_2,r_3\}$, cf., Line~\ref{line:primintr1}.
Further, each row contains a level mapping among atoms in the interpretation such that the corresponding set of proven atoms is obtained, cf., Line~\ref{line:primintr2}.
Row~4 of~$\tab{4}$ for example requires a level mapping~$\sigma_{4.4}$ with~$\sigma_{4.4}(d)<\sigma_{4.4}(a)$ for~$a$ to be proven.
Then, node~$t_5$ is a forget node forgetting~$a$, which keeps only rows, where either~$a$ is not in the interpretation or~$a$ is in the set of proven atoms, and removes~$a$ from the result.
The result of Line~\ref{line:primrem} on~$t_5$ is displayed in table~$\tab{5}$, where Row~3 of~$\tab{4}$ does not have a successor in~$\tab{5}$ since~$a$ is not proven.
For leaf node~$t_6$ we have~$\tab{t_6}=\tab{t_1}$.
Similarly to before, $t_7,t_8$, and~$t_9$ are introduce nodes
and~$\tab{9}$ depicts the resulting table for~$t_9$.
Table~$\tab{10}$ does not contain any successor row of Row~2 of~$\tab{9}$,
since~$c$ is not proven.
Node~$t_{11}$ is a join node combining rows of~$\tab{5}$ and~$\tab{9}$ as given by Line~\ref{line:primjoin}.
Observe that Row~3 of~$\tab{5}$ does not match with any row in~$\tab{9}$. Further, combining Row~3 of~$\tab{5}$ with Row~3 of~$\tab{9}$ results in Row~4 of~$\tab{11}$ (since~$\ell{-}2=3$).
The remaining tables can be obtained similarly. Table~$\tab{16}$ for the root node only depicts (solution) rows, where each atom is proven. 
\end{example}

In contrast to existing work~\cite{FichteHecher19},
if largest SCC size~$\ell< k$, where~$k$ is the treewidth of primal graph~$G_\Pi$,
our algorithm runs in time better than the lower bound given by Proposition~\ref{prop:lowerbound}.
Further, existing work~\cite{FichteHecher19} does not precisely characterize answer sets,
but algorithm \PRIM of Listing~\ref{fig:prim} exactly computes all the answer sets of~$\Pi$.
Intuitively, the reason for this is that level mappings for an atom~$x\in\at(\Pi)$ do not differ in different
bags of~$\mathcal{T}$, but instead we use the same level (at most~$\ell_{\scc(x)}$ many possibilities) 
for~$x$ in all bags. 
Notably, capturing all the answer sets of~$\Pi$ allows that \PRIM can be slightly extended 
to count the answer sets of~$\Pi$ by extending the rows by an integer for counting accordingly.
This can be extended further by enumerating
all the answer sets with linear delay. The resulting enumeration 
algorithm is an anytime algorithm
and just keeps for each row of a table predecessor rows.

\subsubsection*{Consequences on Correctness and Runtime.}

Next, we sketch correctness, which finally allows us to show Theorem~\ref{theo:runtimescc}.

\begin{lemma}[Correctness]\label{thm:algo:correct}
Let~$\Pi$ be a HCF program, where the treewidth of~$G_\Pi$ is at most~$k$ and where
every SCC~$C$ satisfies $|C|{+}1\leq \ell$. 
Then, for a given tree decomposition~$\mathcal{T}=(T,\chi)$ of primal graph $G_\Pi$, algorithm~$\PRIM$ executed for each node~$t$ of~$T$ in post-order is correct.
\end{lemma}
\begin{proof}[Proof (Sketch)]
The proof consists of both soundness, which shows that only correct data is in the tables, and 
completeness saying that no row of any table is missing.
Soundness is established by showing an invariant for every node~$t$, where
the invariant is assumed for every child node of~$t$.
For the invariant, we use auxiliary notation \emph{program~$\Pi_{<t}$ strictly below~$t$}
consisting of~$\Pi_{t'}$ for any node~$t'$ below~$t$, as well as the \emph{program~$\Pi_{\leq t}$ below~$t$}, where~$\Pi_{\leq t}\eqdef \Pi_{< t}\cup \Pi_t$.
Intuitively, this invariant for~$t$ states that every row~$\langle I, \mathcal{P}, \sigma\rangle$
of table~$\tab{t}$ ensures (1) ``satisfiability'': $I\models \Pi_t$, (2)``answer set extendability'': $I$ can be extended to an answer set of~$\Pi_{< t}$, (3)``provability'': $a\in\mathcal{P}$ if and only if 
there is a rule in~$\Pi_{\leq t}$ proving~$a$ with~$\sigma$, and (4)``minimality'':
there is no~$a\in\mathcal{P}, r\in\Pi_{\leq t}$ such that~$r$ proves~$a$
with~$\sigma'$, where~$\sigma'$ coincides with~$\sigma$, but sets~$\sigma'(a)=\sigma(a)-1$.
Notably, the invariant for the empty root node~$n=\rootOf(T)$ ensures that if~$\tab{n}\neq\emptyset$,
there is an answer set of~$\Pi$.
Completeness can be shown by establishing that if~$\tab{t}$ is complete,
then every potential row that fulfills the invariant for any child node~$t'$ of~$t$,
is indeed present in the corresponding table~$\tab{t'}$.
\end{proof}

\begin{restatetheorem}[theo:runtimescc]
\begin{theorem}[Runtime of SCC-bounded \ASP]
Assume a HCF logic program~$\Pi$,
where the treewidth of the primal graph~$G_\Pi$ of~$\Pi$ is at most~$k$
and~$\ell$ is the largest SCC size.
Then, there is an algorithm for deciding the consistency of~$\Pi$,
running in time~$2^{\mathcal{O}(k\cdot \log(\lambda))} \cdot \poly(\Card{\at(\Pi)})$,
where~$\lambda=\min(\{k,\ell\})$. 
\end{theorem}
\end{restatetheorem}
\begin{proof}[Proof] 
First, we compute~\cite{BodlaenderEtAl13} a tree decomposition~$\mathcal{T}=(T,\chi)$
of~$G_\Pi$ that is a 
5-approximation of~$k=\tw{G_\Pi}$ and has a linear number of nodes,
in time~$2^{\mathcal{O}(k)}\cdot \poly(\Card{\at(\Pi)})$.
%
%
Computing~$\ell_{\scc(a)}$ for each atom~$a\in\at(\Pi)$ can be done
in polynomial time.
If~$\ell >k$, we directly run an algorithm~\cite{FichteHecher19} for the consistency of~$\Pi$.
Otherwise, i.e., if~$\ell \leq k$ we run Listing~\ref{fig:prim}
on each node~$t$ of~$T$ in a bottom-up (post-order) traversal.
In both cases, we obtain a total runtime of~$2^{\mathcal{O}(k\cdot \log(\lambda))} \cdot \poly(\Card{\at(\Pi)})$.
\end{proof}

\section{Treewidth-Aware Reductions for SCC-bounded \ASP}

Next, we present a novel reduction from HCF \ASP to tight \ASP.
Given a head-cycle-free logic program,
we present a treewidth-aware reduction that constructs a tight logic program
with little overhead in terms of treewidth.
Concretely, if each SCC of the given head-cycle-free logic program~$\Pi$ has at most~$\ell{-}1$
atoms, the resulting tight program has treewidth~$\mathcal{O}(k\cdot\log(\ell))$.
In the course of this section, we establish the following theorem.

\begin{theorem}[Removing Cyclicity of SCC-bounded \ASP]\label{thm:asptotight}
Let~$\Pi$ be a HCF program, where the treewidth of~$G_\Pi$ is at most~$k$ and where
every SCC~$C$ satisfies $|C|+1\leq \ell$. 
Then, there is a tight program~$\Pi'$ with treewidth in~$\mathcal{O}(k\cdot\log(\ell))$ such that 
the stable models of~$\Pi$ and~$\Pi'$ projected to the atoms of~$\Pi$ coincide.
\end{theorem}

%

\subsection{Reduction to tight \ASP}

The overall construction of the reduction is inspired by the idea of treewidth-aware reductions~\cite{Hecher20}, where in the following, we assume an SCC-bounded program~$\Pi$ and a tree decomposition~$\mathcal{T}=(T,\chi)$ of~$G_\Pi$ such that the construction of the resulting tight logic program~$\Pi'$ is heavily guided along~$\mathcal{T}$.
In contrast to existing work~\cite{Hecher20}, bounding cycles with the largest SCC size additionally allows to have a ``global'' level mapping~\cite{Janhunen06}, i.e., we do not have different levels for an atom in different bags.
Then, while the overall reduction is still guided along the tree decomposition~$\mathcal{T}$ in order to take care to not increase treewidth too much,
these global level mappings ensure that the tight program is 
guaranteed to preserve all answer sets (projected to the atoms of~$\Pi$), 
as stated in Theorem~\ref{thm:asptotight}.
Before we discuss the construction in detail, we require auxiliary atoms and notation as follows.
In order to \emph{guide} the evaluation of the provability of an atom~$x\in\at(\Pi)$ in a node~$t$ in~$T$ along the decomposition~$\mathcal{T}$, we use atoms~$p_t^x$ and~$p_{\leq t}^x$ to indicate that~$x$ was proven in node~$t$ (with some rule in~$\Pi_t$) and below~$t$, respectively.
Further, we require atoms~$b_x^j$, called \emph{level bits}, for~$x\in\at(\Pi)$ and~$1\leq j\leq \ceil{\log(\ell_{\scc(x)})}$,
which are used as bits in order to represent in a level mapping the level of~$x$ in binary.
To this end, we denote for~$x$ and a number~$i$ with~$0 \leq i < \ell_{\scc(x)}$ 
as well as a position number~$1\leq j\leq \ceil{\log(\ell_{\scc(x)})}$,
the~\emph{$j$-th position of~$i$ in binary} by~$\bval{i}{j}$.
Then, we let~$\bvali{x}{t}{i}$ be the consistent \emph{set of literals
over level bits}~$b_x^j$ that is used to represent level number~$i$ for~$x$ in binary.
More precisely, for each position number~$j$, $\bvali{x}{t}{i}$ contains~$b_x^j$ if~$\bval{i}{j}=1$ and $\neg b_x^j$ otherwise, i.e., if~$\bval{i}{j}=0$.
Finally, we also use auxiliary atoms of the form~$x\prec i$ to indicate that the level for~$x$
represented by~$\bvali{x}{t}{i}$ is indeed smaller than~$i>0$.
%
%
%
\begin{example}
Recall program~$\Pi$, level mapping~$\sigma$, and largest SCC size~$\ell=5$ from Example~\ref{ex:running1}.
For representing~$\sigma$ in binary, we require~$\ceil{\log(\ell{-}1)}=2$
bits per atom~$a\in\at(\Pi)$ and we assume that bits are ordered from least to most significant bit.
So $\bval{\sigma(e)}{0}=\bval{\sigma(e)}{1}=0$, $\bval{\sigma(c)}{0}=1$ and~$\bval{\sigma(c)}{1}=0$.
Then, we have~$\bvali{e}{t}{\sigma(e)}=\{\neg b_e^0, \neg b_e^1\}$,
$\bvali{b}{t}{\sigma(b)}=\{\neg b_b^0, \neg b_b^1\}$,
$\bvali{c}{t}{\sigma(c)}=\{b_c^0, \neg b_c^1\}$,
$\bvali{d}{t}{\sigma(d)}=\{\neg b_d^0, b_d^1\}$, and
$\bvali{a}{t}{\sigma(a)}=\{b_a^0, b_a^1\}$.
\end{example}

Next, we are ready to discuss the treewidth-aware reduction from SCC-bounded \ASP to tight \ASP, which takes~$\Pi$ and~$\mathcal{T}$ and creates a tight logic program~$\Pi'$.
To this end, let~$t$ be any node of~$T$.
First, truth values for each atom~$x\in\chi(t)$ are subject to a guess by Rules~(\ref{asp:guessx})
and by Rules~(\ref{asp:check}) it is ensured that all rules of~$\Pi_t$ are satisfied.
Notably, by the definition of tree decompositions, Rules~(\ref{asp:guessx}) and Rules~(\ref{asp:check}) indeed cover all the atoms of~$\Pi$ and all rules of~$\Pi$, respectively.
Then, the next block of rules consisting of Rules~(\ref{asp:guessx2})--(\ref{asp:checkremove2}) is used for ensuring provability and finally the last block of Rules~(\ref{asp:minimality})--(\ref{asp:minimalityproven2}) is required in order to preserve answer sets, i.e., these rules prevent duplicate answer sets of~$\Pi'$ for one specific answer set of~$\Pi$.

For the block of Rules~(\ref{asp:guessx2})--(\ref{asp:checkremove2}) to ensure provability,
we need to guess the level bits for each atom~$x$ as given in Rules~(\ref{asp:guessx2}).
Rules~(\ref{asp:auxbinary}) ensure that we correctly define~$x\prec i$, which is the case if there exists a bit~$\bval{i}{j}$ that is set to~$1$, but we have~$\neg b^j_x$ and for all larger bits~$\bval{i}{j'}$ that are set to~$0$ ($j'>j$), we also have~$\neg b^{j'}_x$.
Then, for Rules~(\ref{asp:checkfirst}) we slightly abuse notation~$x\prec i$ and
use it also for a set~$X$, where~$X\prec i$ denotes a set of atoms of the form~$x\prec i$
for each~$x\in X$. Rules~(\ref{asp:checkfirst}) make sure that whenever a rule~$r\in\Pi_t$ proves~$x$ with the level mapping given by the level bits over atoms in~$\chi(t)$, 
we have provability~$p_t^x$ for~$x$ in~$t$.
However, only for the atoms of the positive body~$B_r^+$ which are also in the same SCC~$C=\scc(x)$ as~$x$ we need to check that the levels are smaller than the level of~$x$, since by definition of SCCs, there cannot be a positive cycle among atoms of different SCCs.
As a result, if there is a rule, where no atom of the positive body is in~$C$, satisfying the rule is enough for proving~$x$ as given by Rules~(\ref{asp:checkfirst2}).
If provability~$p_t^x$ holds, we also have~$p_{\leq t}^x$ by Rules~(\ref{asp:prove}) 
and provability is propagated from node~$t'$ to its parent node~$t$ by setting~$p_{\leq t}$
if~$p_{\leq t'}$, as indicated by Rules~(\ref{asp:prove2}).
Finally, whenever an atom~$x$ is forgotten in a node~$t$, we require to have provability~$p_{\leq t}^x$ ensured by Rules~(\ref{asp:checkremove}) and~(\ref{asp:checkremove2}) since~$t$ might be~$\rootOf(T)$.

\emph{Preserving answer sets:} The last block consisting of Rules~(\ref{asp:minimality}), (\ref{asp:minimalityproven}), and~(\ref{asp:minimalityproven2}) makes sure that atoms that are false or not in the answer set of~$\Pi'$ get level~$0$ and that
we do prohibit levels for an atom~$x$ that can be safely decreased by one without
loosing provability. This ensures that for each answer set of~$\Pi$ we get exactly one corresponding answer set of~$\Pi'$ and vice versa.

{
\vspace{-1em}
\begin{align}
%
\label{asp:guessx}&\{x\} \hsep&&{\text{for each }x\in\chi(t)\text{; see}^{\ref{foot:choice}}}\\
\label{asp:check}&\leftarrow B_r^+, \overline{B_r^-\cup H_r}&&{\text{for each }r\in\Pi_t}\\[1.5em]
\label{asp:guessx2}&\{b_x^j\} \hsep&&{\text{for each }x\in\chi(t),}\notag\\[-.5em]
&&&1\leq j\leq \ceil{\log(\ell_{\scc(x)})}\text{; see}^{\ref{foot:choice}}\\
\label{asp:auxbinary}	& x \prec i \leftarrow \neg b_x^j, \neg b_x^{j_1}, \ldots, \neg b_x^{j_s}&& \text{for each }x\in\chi(t), C=\scc(x),\notag\\[-.5em]
&&&  1 \leq i < \ell_{C}, 1 \leq j \leq \ceil{\log(\ell_{C})}, \notag\\[-.5em]
&&& \bval{i}{j}{=}1,\{j' \mid j < j' \leq\ceil{\log(\ell_{C})},\notag\\[-.5em]
&&& \bval{i}{j'}{=}0\}=\{j_1,\ldots, j_s\}\\
%
%
%
	\label{asp:checkfirst}&p_{t}^x \leftarrow x, \bvali{x}{t}{i}, (B_r^+{\cap} C) {\prec} i,  B_r^+, \overline{B_r^-{\cup}(H_r{\setminus}\{x\})} &&{\text{for each } r\in\Pi_t,x\in \chi(t)\text{ with }}\notag\\[-.5em]
&&&x\in H_r, C=\scc(x),1\leq i < \ell_C,\notag\\[-.5em]
&&&\text{and } B_r^+\cap C\neq\emptyset\\ 
	\label{asp:checkfirst2}&p_{t}^x \leftarrow x,  B_r^+, \overline{B_r^-{\cup}(H_r{\setminus}\{x\})} &&{\text{for each } r\in\Pi_t,x\in \chi(t)\text{ with }}\notag\\[-.5em]
&&&x\in H_r, B_r^+\cap\scc(x)=\emptyset\\
%
%
%
%
	\label{asp:prove}&p^x_{\leq t} \hsep p^x_t&&{\text{for each } x\in \chi(t)}\\
\label{asp:prove2}&p^x_{\leq t} \hsep p^x_{\leq t'}&&{\text{for each } x\in \chi(t),}\notag\\[-.5em]
&&& t'\in\children(t), x\in\chi(t')\\
\label{asp:checkremove}&\hsep x, \neg p^{x}_{\leq t'}&&{\text{for each }{t'\in\children(t)},}\notag\\[-.45em]
	&&& x\in\chi(t')\setminus\chi(t)\\
	\label{asp:checkremove2}&\hsep x, \neg p^{x}_{\leq n}&&{\text{for each }x\in\chi(n),}\notag\\[-.35em]&&&n=\rootOf(T)\\[1.5em]
\label{asp:minimality}&\leftarrow \neg x, b^j_x&&\text{for each }x\in\chi(t),\notag\\[-.5em]
&&&1\leq j\leq \ceil{\log(\ell_{\scc(x)})}\\
\label{asp:minimalityproven}&\leftarrow x,\bvali{x}{t}{i}, (B_r^+{\cap} C) {\prec} i{-}1,  B_r^+, \overline{B_r^-{\cup}(H_r{\setminus}\{x\})}&&{\text{for each } r\in\Pi_t,x\in \chi(t)\text{ with }}\notag\\[-.5em]
&&&x\in H_r, C=\scc(x),2\leq i < \ell_C,\notag\\[-.5em]
&&&\text{and } B_r^+\cap C\neq\emptyset\\
	%
	%
	\label{asp:minimalityproven2}&\leftarrow x,\bvali{x}{t}{i}, B_r^+, \overline{B_r^-{\cup}(H_r{\setminus}\{x\})}&&{\text{for each } r\in\Pi_t,x\in \chi(t)\text{ with }}\notag\\[-.5em]
&&&x\in H_r, C=\scc(x), 1\leq i< \ell_C, \notag\\[-.5em]
&&&\text{and }B_r^+\cap C=\emptyset
\end{align}
\vspace{-.05em}
}

\footnoteitext{\label{foot:choice}A \emph{choice rule}~\cite{SimonsNiemelaeSoininen02} is of the form $\{a\} \hsep$ and in a HCF logic program it corresponds to a disjunctive rule $a \lor a' \hsep$, where $a'$ is a fresh atom.}

\begin{example}
Recall program~$\Pi$ of Example~\ref{ex:running1} and TD~$\mathcal{T}=(T,\chi)$ of~$G_\Pi$ as given in Figure~\ref{fig:graph-td}.
Rules~(\ref{asp:guessx}) and Rules~(\ref{asp:check}) are constructed for each atom~$a\in\at(\Pi)$ and for each rule~$r\in\Pi$, respectively.
Similarly, Rules~(\ref{asp:guessx2}) are constructed for each of the~$\ceil{\log(\ell_{a})}$ many bits of each atom~$a\in\at(\Pi)$.
Rules~(\ref{asp:auxbinary}) serve as auxiliary definition, where for, e.g., atom~$c$
we construct~$c{\prec}1 \leftarrow \neg b_c^0, \neg b_c^1$; 
$c{\prec}2 \leftarrow \neg b_c^1$;
$c{\prec}3 \leftarrow \neg b_c^0$; and 
$c{\prec}3 \leftarrow \neg b_c^1$.
Next, we show Rules~(\ref{asp:checkfirst})--(\ref{asp:minimalityproven2}) for node~$t_2$ of~$T$.

\noindent\begin{tabular}{@{\hspace{0.15em}}l@{\hspace{0.15em}}|@{\hspace{0.15em}}l@{\hspace{0.0em}}}
Rule number & Rules\\
\hline
(\ref{asp:checkfirst})& $p_{t_2}^b \leftarrow b, \bvali{b}{}{1}, d{\prec}1, d$; $p_{t_2}^b \leftarrow b, \bvali{b}{}{2}, d{\prec}2, d$; $p_{t_2}^b \leftarrow b, \bvali{b}{}{3}, d{\prec}3, d$;\\
& $p_{t_2}^c \leftarrow c, \bvali{c}{}{1}, d{\prec}1, d$; $p_{t_2}^c \leftarrow c, \bvali{c}{}{2}, d{\prec}2, d$; $p_{t_2}^c \leftarrow c, \bvali{c}{}{3}, d{\prec}3, d$;\\
& $p_{t_2}^d \leftarrow d, \bvali{d}{}{1}, b{\prec}1, c{\prec}1, b,c$; $p_{t_2}^d \leftarrow d, \bvali{d}{}{2}, b{\prec}2, c{\prec}2, b,c$;\\& $p_{t_2}^d \leftarrow d, \bvali{d}{}{3}, b{\prec}3, c{\prec}3, b,c$\\
(\ref{asp:prove})& $p_{\leq t_2}^b\leftarrow p_{t_2}^b$; $p_{\leq t_2}^c\leftarrow p_{t_2}^c$; $p_{\leq t_2}^d\leftarrow p_{t_2}^d$\\
%
%
(\ref{asp:minimality}) & $\leftarrow \neg b, b^0_b$; $\leftarrow \neg b, b^1_b$; $\leftarrow \neg c, b^0_c$; $\leftarrow \neg c, b^1_c$; $\leftarrow \neg d, b^0_d$; $\leftarrow \neg d, b^1_d$\\ 
(\ref{asp:minimalityproven}) & $\leftarrow b, \bvali{b}{}{2}, d{\prec}1, d$; $\leftarrow b, \bvali{b}{}{3}, d{\prec}2, d$;\\
& $\leftarrow c, \bvali{c}{}{2}, d{\prec}1, d$; $\leftarrow c, \bvali{c}{}{3}, d{\prec}2, d$;\\
& $\leftarrow d, \bvali{d}{}{2}, b{\prec}1, c{\prec}1, b,c$; 
 $\leftarrow d, \bvali{d}{}{3}, b{\prec}2, c{\prec}2, b,c$\\
(\ref{asp:checkfirst2}),(\ref{asp:prove2})--(\ref{asp:checkremove2}),(\ref{asp:minimalityproven2})&-\\
\end{tabular}%

\smallskip
For the root node~$t_5$ of~$T$, we obtain the following Rules~(\ref{asp:checkfirst})--(\ref{asp:minimalityproven2}).

\noindent\begin{tabular}{@{\hspace{0.15em}}l@{\hspace{0.15em}}|@{\hspace{0.15em}}l@{\hspace{0.0em}}}
Rule number & Rules\\
\hline
(\ref{asp:checkfirst2})& $p_{t_5}^b \leftarrow b, e, \neg f$\\
(\ref{asp:prove})& $p_{\leq t_5}^b\leftarrow p_{t_5}^b$; $p_{\leq t_5}^e\leftarrow p_{t_5}^e$; $p_{\leq t_5}^f\leftarrow p_{t_5}^f$\\
%
(\ref{asp:prove2})& $p_{\leq t_5}^b \leftarrow p_{\leq t_3}^b$; $p_{\leq t_5}^e \leftarrow p_{\leq t_3}^e$; $p_{\leq t_5}^e \leftarrow p_{\leq t_4}^e$; $p_{\leq t_5}^f \leftarrow p_{\leq t_4}^f$\\
(\ref{asp:checkremove}) & $\leftarrow d, \neg p_{\leq t_3}^d$; $\leftarrow g, \neg p_{\leq t_4}^g$\\
(\ref{asp:checkremove2}) & $\leftarrow b, \neg p_{\leq t_5}^b$; $\leftarrow e, \neg p_{\leq t_5}^e$; $\leftarrow f, \neg p_{\leq t_5}^f$\\
(\ref{asp:minimality}) & $\leftarrow \neg b, b^0_b$; $\leftarrow \neg b, b^1_b$; $\leftarrow \neg e, b^0_e$; $\leftarrow \neg e, b^1_e$; $\leftarrow \neg f, b^0_f$; $\leftarrow \neg f, b^1_f$\\ 
(\ref{asp:minimalityproven2}) & $\leftarrow b, \bvali{b}{}{1}, e, \neg f$; $\leftarrow b, \bvali{b}{}{2}, e, \neg f$; $\leftarrow b, \bvali{b}{}{3},e, \neg f$\\
(\ref{asp:checkfirst}),(\ref{asp:minimalityproven})&-
\end{tabular}%
\end{example}

\subsubsection*{Correctness and Treewidth-Awareness.} 

\begin{lemma}[Correctness]\label{thm:red:correct}
Let~$\Pi$ be a HCF program, where the treewidth of~$G_\Pi$ is at most~$k$ and where
every SCC~$C$ satisfies $|C|{+}1\leq \ell$. 
Then, the tight program~$\Pi'$ obtained by the reduction above on~$\Pi$ and a tree decomposition~$\mathcal{T}=(T,\chi)$ of primal graph $G_\Pi$, is correct.
Formally, for any answer set~$I$ of~$\Pi$ there is exactly one answer set~$I'$ of~$\Pi'$ as given by Rules~(\ref{asp:guessx})--(\ref{asp:minimalityproven2}) and vice versa.
\end{lemma}
\begin{proof}
``$\Longrightarrow$'': Given any answer set~$I$ of~$\Pi$. Then, there exists a unique~\cite{Janhunen06}, 
minimal level mapping~$\sigma$ proving each~$x\in I$ with~$0\leq \sigma(x) < \ell_{\scc(x)}$.
Let~$P\eqdef\{p_t^x, p_{\leq t}^x \mid r\in\Pi_t \text{ proves } x \text{ with }\sigma, x\in I, t \text{ in }T\}$.
From this we construct an interpretation~$I'\eqdef I\cup\{b^j_x \mid \bval{\sigma(x)}{j} = 1, 0 \leq j \leq \ceil{\log(\ell_{\scc(x)})}, x\in I\}\cup P \cup \{ p_{\leq t}^x \mid x\in I, t' \in T, t' \text{ is below }t\text{ in }T, p_{\leq t'}^x\in P\}$, which sets atoms as~$I$ and additionally encodes~$\sigma$ in binary
and sets provability accordingly. It is easy to see that~$I'$ is an answer set of~$\Pi'$.
``$\Longleftarrow$'': Given any answer set~$I'$ of~$\Pi'$.
From this we construct~$I\eqdef I'\cap\at(\Pi)$ as well as level 
mapping~$\sigma \eqdef \{x\mapsto f_I(x) \mid x\in\at(\Pi)\}$, where we define function~$f_{I'}(x): \at(\Pi)\rightarrow \{0,\ldots,\ell{-}2\}$ for atom~$x\in\at(\Pi)$ to return~$1\leq 0 < \ell_{\scc(x)}$ if $\{b_x^j \mid 0 \leq j \leq \ceil{\log(\ell_{\scc(x)})}, \bval{i}{j}=1\}= \{b_x^j \in I' \mid 0 \leq j \leq \ceil{\log(\ell_{\scc(x)})}\}$,
i.e., the atoms in answer set~$I'$ binary-encode~$i$ for~$x$.
Assume towards a contradiction that~$I\not\models\Pi$. But then~$I'$ does not satisfy at least one instance of Rules~(\ref{asp:guessx}) and~(\ref{asp:check}), contradicting that~$I'$ is an answer set of~$\Pi'$.
Again, towards a contradiction assume that~$I$ is not an answer set of~$\Pi$,
i.e., at least one~$x\in\at(\Pi)$ cannot be proven with~$\sigma$.
Then, we still have~$p_{\leq n}^x\in I'$ for~$n=\rootOf(T)$, by Rules~(\ref{asp:checkremove}) and~(\ref{asp:checkremove2}).
However, then we either have that~$p_{\leq t}^x\in I'$ or~$p_{n}^x\in I'$ by Rules~(\ref{asp:prove}) and~(\ref{asp:prove2}) for at least one child node~$t$ of~$n$.
Finally, by the connectedness property (iii) of the definition of TDs, we have that there has to be a
node~$t'$ that is either~$n$ or a descendant of~$n$ where we have~$p_{t'}^x\in I'$.
Consequently, by Rules~(\ref{asp:checkfirst}) and~(\ref{asp:checkfirst2}) as well as auxiliary Rules~(\ref{asp:guessx2}) and~(\ref{asp:auxbinary}) we have that there is a rule~$r\in \Pi$ that proves~$x$ with~$\sigma$, contradicting the assumption.
Similarly, one can show that Rules~(\ref{asp:minimality}) and~(\ref{asp:minimalityproven}),(\ref{asp:minimalityproven2}) ensure minimality of~$\sigma$.
\end{proof}

\begin{lemma}[Treewidth-Awareness]\label{thm:red:twaware}
Let~$\Pi$ be a HCF program, where the treewidth of~$G_\Pi$ is at most~$k$ and where
every SCC~$C$ satisfies $|C|{+}1\leq \ell$. 
Then, the treewidth of tight program~$\Pi'$ obtained by the reduction above by using~$\Pi$ and a tree decomposition~$\mathcal{T}=(T,\chi)$ of primal graph $G_\Pi$, is in~$\mathcal{O}(k\cdot \log(\ell))$.
\end{lemma}
\begin{proof}[Proof (Sketch)]
We take~$\mathcal{T}=(T,\chi)$ and construct a TD~$\mathcal{T}=(T,\chi')$ of~$G_{\Pi'}$,
where~$\chi'$ is defined as follows.
For every node~$t$ of~$T$, whose parent node is~$t^*$, we let~$\chi'(t)\eqdef \chi(t)\cup\{b_x^j\mid x\in\chi(t), 0 \leq j \leq \ceil{\log(\ell_{\scc(x)})}\} \cup \{p_t^x, p_{\leq t}^x, p_{\leq t^*} \mid x\in\chi(t)\}$.
It is easy to see that indeed all atoms of every instance of Rules~(\ref{asp:guessx})--(\ref{asp:minimalityproven2}) appear in at least one common bag of~$\chi'$. Further, we also have connectedness of~$\mathcal{T}'$, i.e., $\mathcal{T}'$ is indeed a well-defined TD of~$G_{\Pi'}$
and~$\Card{\chi(t)}$ in $\mathcal{O}(k\cdot\log(\ell))$.
\end{proof}

\longversion{
\begin{lemma}[Runtime]\label{thm:red:runtime}
Let~$\Pi$ be a HCF program, where the treewidth of~$G_\Pi$ is at most~$k$ and where
every SCC~$C$ satisfies $|C|{+}1\leq \ell$. 
Then, for a given tree decomposition~$\mathcal{T}$ of primal graph $G_\Pi$, the reduction above on~$\Pi$ and~$\mathcal{T}$ runs in time~$\mathcal{O}(k \cdot \ceil{\log(\ell)}\cdot \poly(\Card{\at(\Pi)})).$
\end{lemma}}

Finally, we are in the position to prove Theorem~\ref{thm:asptotight} by combining both lemmas.

\begin{restatetheorem}[thm:asptotight]
\begin{theorem}[Removing Cyclicity of SCC-bounded \ASP]
Let~$\Pi$ be a HCF program, where the treewidth of~$G_\Pi$ is at most~$k$ and where
every SCC~$C$ satisfies $|C|{+}1\leq \ell$. 
Then, there is a tight program~$\Pi'$ with treewidth in~$\mathcal{O}(k\cdot\log(\ell))$ such that 
the stable models of~$\Pi$ and~$\Pi'$ projected to the atoms of~$\Pi$ coincide.
\end{theorem}
\end{restatetheorem}
\begin{proof}
First, we compute a tree decomposition~$\mathcal{T}=(T,\chi)$
of~$G_\Pi$ that is a 
5-approximation of~$k=\tw{G_\Pi}$ in time~$2^{\mathcal{O}(k)}\cdot \poly(\Card{\at(\Pi)})$.
Observe that the reduction consisting of Rules~(\ref{asp:guessx})--(\ref{asp:minimalityproven2}) on~$\Pi$ and~$\mathcal{T}$
runs in polynomial time, precisely in time~$\mathcal{O}(k \cdot \log(\ell)\cdot \poly(\Card{\at(\Pi)}))$.
The claim follows by correctness (Lemma~\ref{thm:red:correct}) and by treewidth-awareness as given by Lemma~\ref{thm:red:twaware}.
\end{proof}

Having established Theorem~\ref{thm:asptotight}, the reduction above easily allows for an alternative proof of Theorem~\ref{theo:runtimescc}.
Instead of Algorithm~\PRIM of Listing~\ref{fig:prim}, one could also compile
the resulting tight program of the reduction above to a propositional formula (\SAT),
and use an existing algorithm for \SAT to decide satisfiability.
Indeed, such algorithms run in time single-exponential in the treewidth~\cite{SamerSzeider10b} and we end up with similar worst-case running times as given by Theorem~\ref{theo:runtimescc}.

\subsection{Reduction to \SAT}
Having established the reduction of SCC-bounded \ASP to tight \ASP,
we now present a treewidth-aware reduction of tight \ASP to \SAT,
which together allow to reduce from SCC-bounded \ASP to \SAT.
While the step from tight \ASP to \SAT might seem straightforward for the program~$\Pi'$
obtained by the reduction above, in general it is not guaranteed
that existing reductions, e.g.~\cite{Fages94,LinZhao03,Janhunen06}, do not cause a significant blowup
in the treewidth of the resulting propositional formula.
Indeed, one needs to take care and define a treewidth-aware reduction.

\futuresketch{
\begin{example}
Existing reductions screw up.
\end{example}}

Let~$\Pi$ be any given tight logic program and~$\mathcal{T}=(T,\chi)$
be a tree decomposition of~$G_\Pi$.
Similar to the reduction from SCC-bounded \ASP to tight \ASP,
we use as variables besides the original atoms of~$\Pi$ also auxiliary variables.
In order to preserve treewidth, we still need to guide the evaluation of the provability of an atom~$x\in\at(\Pi)$ in a node~$t$ in~$T$ along the TD~$\mathcal{T}$, whereby we use atoms~$p_t^x$ and~$p_{\leq t}^x$ to indicate that~$x$ was proven in node~$t$ and below~$t$, respectively.
However, we do not need any level mappings, since there is no positive cycle in~$\Pi$,
but we still guide the idea of Clark's completion~\cite{Clark77} along TD~$\mathcal{T}$.
Consequently, we construct the following propositional formula, where for each node~$t$ of~$T$ we add Formulas~(\ref{red:checkrules})--(\ref{red:checkfirst}).
Intuitively, Formulas~(\ref{red:checkrules}) ensure that all rules are satisfied, cf., Rules~(\ref{asp:check}).
Formulas~(\ref{red:checkremove}) and~(\ref{red:checkremove2}) take care that ultimately an
atom that is set to true requires to be proven, similar to Rules~(\ref{asp:checkremove}) and~(\ref{asp:checkremove2}).
Finally, Formulas~(\ref{red:checkfirst}) and (\ref{red:check}) provide the definition
for an atom to be proven in a node and below a node, respectively, which is similar to
Rules~(\ref{asp:checkfirst})--(\ref{asp:prove2}), but without the level mappings.

\emph{Preserving answer sets:} Answer sets are already preserved, i.e., we obtain exactly one model of the resulting propositional formula~$F$ for each answer set of~$\Pi$ and vice versa. If the equivalence ($\leftrightarrow$) in Formulas~(\ref{red:checkfirst}) and (\ref{red:check}) is replaced
by an implication ($\rightarrow$), we might get duplicate models for one answer set while still ensuring preservation of consistency, i.e., the answers to both decision problems coincide.

{
\vspace{-1em}
\begin{align}
	\label{red:checkrules}&\bigvee_{a\in B_r^+} \neg a \vee \bigvee_{a\in B_r^- \cup H_r} a &&{\text{for each } r\in \Pi_t}\\
	\label{red:checkremove}&x\rightarrow p^{x}_{\leq t'}&&{\text{for each }{t'\in\children(t)},}\notag\\[-.45em]
	&&& x\in\chi(t')\setminus\chi(t)\\
	\label{red:checkremove2}&x \rightarrow p^{x}_{\leq n}&&{\text{for each }x\in\chi(n),}\notag\\[-.35em]&&&n=\rootOf(T)\\ 
%
	\label{red:checkfirst}&p_{t}^x \leftrightarrow \bigvee_{r\in\Pi_t, x \in H_r} (\hspace{-0.15em}\bigwedge_{a\in B_r^+}\hspace{-.5em}a \wedge x  \wedge\hspace{-1.75em} \bigwedge_{b\in B_r^- \cup (H_r \setminus \{x\})}\hspace{-2.25em}\neg b) &&{\text{for each } x\in \chi(t)}\\
	\label{red:check}&p^x_{\leq t} \leftrightarrow p^x_t \vee (\hspace{-1em}\bigvee_{t' \in \children(t), x\in\chi(t')} \hspace{-2em} p^x_{\leq t'})&&{\text{for each } x\in \chi(t)}
\end{align}
\vspace{-.05em}
}
\longversion{
\begin{example}
also mention that in our case, obviously, the completion is already done.
in other words, while in general this is not so obvious, in our case this part has already been done by the reduction to tight asp.
\end{example}}	

\subsubsection*{Correctness and Treewidth-Awareness.} 

Conceptually the proofs of the next two Lemmas~\ref{thm:red2:correct} and~\ref{thm:red2:twaware}
proceed rather similar to the proofs of Lemmas~\ref{thm:red:correct} and~\ref{thm:red:twaware}, but without
the level mappings, respectively.

\begin{lemma}[Correctness]\label{thm:red2:correct}
Let~$\Pi$ be a tight logic program, where the treewidth of~$G_\Pi$ is at most~$k$.
Then, the propositional formula~$F$ obtained by the reduction above on~$\Pi$ and a tree decomposition~$\mathcal{T}$ of primal graph $G_\Pi$, consisting of Formulas~(\ref{red:checkrules})--(\ref{red:check}), is correct.
Formally, for any answer set~$I$ of~$\Pi$ there is exactly one satisfying assignment of~$F$ and vice versa.
\end{lemma}

\begin{lemma}[Treewidth-Awareness]\label{thm:red2:twaware}
Let~$\Pi$ be a tight logic program, where the treewidth of~$G_\Pi$ is at most~$k$.
Then, the treewidth of propositional formula~$F$ obtained by the reduction above by using~$\Pi$ and a tree decomposition~$\mathcal{T}$ of primal graph $G_\Pi$, is in~$\mathcal{O}(k)$.
\end{lemma}
\begin{proof}
The proof proceeds similar to Lemma~\ref{thm:red:twaware}.
However, due to Formulas~(\ref{red:check}) and
without loss of generality one needs to consider only TDs, where every node has a constant number of child nodes. Such a TD can be easily obtained from any given TD by adding auxiliary nodes accordingly~\cite{Kloks94a}.
\end{proof}

\longversion{
\begin{lemma}[Runtime]\label{thm:red2:runtime}
Let~$\Pi$ be a tight logic program, where the treewidth of~$G_\Pi$ is at most~$k$.
Then, for a given tree decomposition~$\mathcal{T}$ of primal graph $G_\Pi$, the reduction above on~$\Pi$ and~$\mathcal{T}$ runs in time~$\mathcal{O}(k \cdot \poly(\Card{\at(\Pi)})).$
\end{lemma}}

However, we cannot do much better, as given by the following proposition.

\begin{proposition}[ETH-Tightness]\label{thm:red2:lowerbound}
Let~$\Pi$ be a tight logic program, where the treewidth of~$G_\Pi$ is at most~$k$.
Then, under ETH, the treewidth of the resulting propositional formula~$F$ can not be significantly improved, i.e., under ETH there is no reduction running in time~$2^{o(k)}\cdot\poly(\Card{\at(\Pi)})$ such that~$\tw{G_F}$ is in~$o(k)$.
\end{proposition}
\begin{proof}
First, we reduce \SAT to tight \ASP, i.e., capture all models of a given formula~$F$ in a tight program~$\Pi$. Thereby~$\Pi$ consists of a choice rule for each variable of~$F$ and a constraint for each clause. 
Towards a contradiction assume the contrary of this proposition.
Then, we reduce~$\Pi$ back to a propositional formula~$F'$, running in time~$2^{o(k)}\cdot\poly(\Card{\at(\Pi)})$ with~$\tw{G_{F'}}$ being in~$o(k)$.
Consequently, we use an algorithm for \SAT~\cite{SamerSzeider10b} on~$F'$ to effectively solve~$F$ in time~$2^{o(k)}\cdot\poly(\Card{n})$, where~$F$ has~$n$ variables, which finally contradicts ETH.
\end{proof}

Knowing that under ETH tight \ASP has roughly the same complexity for treewidth as \SAT, we can derive the following corollary that complements the existing lower bound for normal ASP as given by Proposition~\ref{prop:lowerbound}.

\begin{corollary}
Let~$\Pi$ be any normal logic program, where the treewidth of~$G_\Pi$ is at most~$k$.
Then, under ETH, there is no reduction to a tight logic program~$\Pi'$ running in time~$2^{o(k\cdot\log(k))}\cdot\poly(\Card{\at(\Pi)})$ such that~$\tw{G_{\Pi'}}$ is in~$o(k\cdot\log(k))$.
\end{corollary}

\section{Conclusion and Future Work}\label{sec:discussion}
This paper deals with improving existing algorithms for deciding 
consistency of head-cycle-free (HCF) ASP for bounded treewidth.
The existing lower bound implies that under the exponential time hypothesis (ETH),
we cannot expect to solve a given HCF program
with~$n$ atoms and treewidth~$k$ in time~$2^{o(k\cdot\log(k))}\cdot\poly(n)$.

In this work, in addition to the treewidth, we also consider the size~$\ell$ of the largest strongly-connected component of the positive dependency graph.
Considering both parameters, we obtain a more precise characterization of the runtime: of~$2^{\mathcal{O}(k\cdot\log(\lambda)}\cdot\poly(n)$, where~$\lambda=\min(\{k,\ell\})$.
This improves the previous result when the strongly-connected components are smaller than the treewidth.
Further, we provide a treewidth-aware reduction from HCF ASP to tight ASP,
where the treewidth increases from~$k$ to~$\mathcal{O}(k\cdot\log(\ell))$.
Finally, we show that under ETH, tight ASP has roughly the same complexity lower bounds
as SAT, which implies that there cannot be a reduction from HCF ASP to tight ASP
such that the treewidth only increases from~$k$ to~$o(k\cdot\log(k))$.

Currently, we are performing experiments and practical analysis of our provided reductions. For future work we suggest to investigate precise lower bounds by considering extensions of ETH like the strong ETH~\cite{ImpagliazzoPaturi01}.
It might be also interesting to establish lower bounds by taking both parameters~$k$ and~$\ell$ into account.
%

\bibliography{include/bibliography/krr,references,include/bibliography/procs}

\end{document}